\documentclass[runningheads]{llncs}
\usepackage[paperheight=235mm,paperwidth=155mm,textwidth=12.2cm,textheight=19.3cm]{geometry}
\usepackage{bbding}
\usepackage{graphicx}
\makeatletter
\RequirePackage[bookmarks,unicode,colorlinks=true]{hyperref}%
   \def\@citecolor{blue}%
   \def\@urlcolor{blue}%
   \def\@linkcolor{blue}%

\def\orcidID#1{\smash{\href{http://orcid.org/#1}{\protect\raisebox{-1.25pt}{\protect\includegraphics{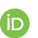}}}}}
\makeatother

\usepackage{todonotes}
\usepackage[utf8]{inputenc}
\usepackage[english]{babel}
\usepackage{amsmath}
\usepackage{amsfonts}
\usepackage{amssymb}
\usepackage{csquotes}
\usepackage{stix}
\usepackage{bussproofs}
\usepackage[ruled,vlined,linesnumbered]{algorithm2e}
\usepackage{listings}
\usepackage{tikz}
\usepackage[strings]{underscore}
\usepackage{cite}
\usepackage{listings}

\newtheorem{fact}{Fact}

\lstdefinelanguage{scala}{
  alsoletter={@,=,>},
  morekeywords={abstract, case, class, def, do, Input, Output, then,
        else, extends, false, free, if, implicit, match,
        object, true, val, var, while, sealed, or,
        for, dependent, null, type, with, try, catch, finally,
        import, final, return, new, override, this, trait,
        private, public, protected, package, throw},
  sensitive=true,
  morecomment=[l]{//},
  morecomment=[s]{/*}{*/},
  morestring=[b]",  
}
\newcommand{\RA}{\Rightarrow}
\newcommand{\starra}[0]{\stackrel{*}{\rightarrow}}
\newcommand{\starlra}[0]{\stackrel{*}{\leftrightarrow}}
\newcommand{\OCBSL}[0]{\mbox{OCBSL}}

\title{Equivalence Checking for Orthocomplemented Bisemilattices in Log-Linear Time\thanks{We acknowledge the financial support of the Swiss National Science Foundation project 200021\_197288
``A Foundational Verifier''. \\
\copyright The Author(s) 2022}}

\author{Simon Guilloud(\Envelope)\orcidID{0000-0001-8179-7549} \and Viktor Kun\v{c}ak\orcidID{0000-0001-7044-9522}}
\institute{EPFL IC LARA, Station 14, CH-1015 Lausanne, Switzerland \\ 
\email{\{Simon.Guilloud,Viktor.Kuncak\}@epfl.ch}}

\authorrunning{S. Guilloud and V. Kun\v{c}ack}

\begin{document}

\maketitle

\begin{abstract}
Motivated by proof checking, we consider the problem of efficiently establishing equivalence of propositional formulas by relaxing the completeness requirements while still providing certain consistency guarantees. To this extent, we present a quasilinear time algorithm to decide the word problem on a natural algebraic structures we call orthocomplemented bisemilattices, a subtheory of Boolean algebra. The starting point for our procedure is a variation of Aho, Hopcroft, Ullman algorithm for isomorphism of trees, which we generalize to directed acyclic graphs. We combine this algorithm with a term rewriting system we introduce to decide equivalence of terms. We prove that our rewriting system is terminating and confluent, implying the existence of a normal form. We then show that our algorithm computes this normal form in log linear (and thus sub-quadratic) time. We provide pseudocode and a minimal working implementation in Scala.
\end{abstract}

\section{Introduction}

Reasoning about propositional logic and its extensions is a basis of many verification algorithms \cite{KroeningStrichman}. Propositional variables may correspond to, for example, sub-formulas in first-order logic theories of SMT
solvers \cite{merzAutomaticVerificationTLA2012, barrettCVC42011, bruttomessoOpenSMTSolver2010a},
hypotheses and lemmas inside proof assistants \cite{wenzelIsabelleFramework2008, harrisonHOLLightOverview2009, naumowiczBriefOverviewMizar2009}, or
abstractions of sets of states.  In particular, it is often of interest to establish that \emph{two propositional formulas are equivalent}.
The equivalence problem for propositional logic is coNP-complete as a negation of propositional satisfiability \cite{Cook10.1145/800157.805047}. 
From proof complexity point of view \cite{ProofComplexityPudlak2019} many known
proof systems, including (non-extended) resolution \cite{Urquhart10.1145/7531.8928} and cutting planes \cite{pudlakLengthsProofs1998} have exponential-sized shortest proofs for certain propositional formulas. SAT and SMT solvers rely on DPLL-style algorithms \cite{DPLL,
ganzingerDPLLFastDecision2004} and do not have
polynomial run-time guarantees on equivalence checking, even if formulas are syntactically close. Proof assistants implement such algorithms as tactics, so they have similar difficulties.
A consequence of this is that implemented systems may take a very long time (or fail to acknowledge) that a large formula is equivalent to its minor variant  differing in, for example, reordering of internal conjuncts or disjuncts. Similar situations also arise in program verifiers \cite{DafnyCalculations2014,AutoActive2015,HamzaETAL19SystemFR,ZeeETAL09IntegratedProofLanguageforImperativePrograms, ZeeETAL08FullFunctionalVerificationofLinkedDataStructures}, where assertions act as lemmas in a proof.  

It is thus natural to ask for an approximation of the propositional equivalence problem:
\emph{can we find an expressive theory supporting many of the algebraic laws of Boolean algebra but for which we can still have a complete and efficient algorithm for formula equivalence?} By efficient, we mean about as fast, up to logarithmic factors, as the simple linear-time syntactic comparison of formula trees.

We can use such an efficient equivalence algorithm to construct more flexible proof systems. Consider any sound proof system for propositional logic and replace the notion of \emph{identical} sub-formulas with our notion of fast equivalence. For example, the axiom schema $p \rightarrow (q \rightarrow p)$ becomes $p \rightarrow (q \rightarrow p')$ for all equivalent $p$ and $p'$. The new system remains sound. It accepts all the previously admissible inference steps, but also some new ones, which
makes it more flexible.

\begin{table}[bth]
    \centering
    \begin{tabular}{r c | r c}
         L1: & $x \sqcup y = y \sqcup x$  & L1': & $x \land y = y \land x$ \\
         L2: & $x \sqcup ( y \sqcup z) = (x \sqcup y) \sqcup z$  & L2': & $x \land ( y \land z) = (x \land y) \land z$ \\
         L3: & $x \sqcup x = x$  & L3': & $x \land x = x$ \\
         L4: & $x \sqcup 1 = 1$  & L4': & $x \land 0 = 0$ \\
         L5: & $x \sqcup 0 = x$  & L5': & $x \land 1 = x$ \\
         L6: & $\neg \neg x = x$  & L6': & same as L6  \\
         L7: & $x \sqcup \neg x = 1$  & L7': & $x \land \neg x = 0$ \\
         L8: & $\neg (x \sqcup y) = \neg x \land \neg y$  & L8': &  $\neg (x \land y) = \neg x \sqcup \neg y$ 
    \end{tabular}
    
    \
    
    \caption{Laws of an algebraic structures $(S, \land, \sqcup, 0, 1, \neg)$.
    Our algorithm is complete (and log-linear time) for structures that satisfy laws L1-L8 and L1'-L8'. We call these structures orthocomplemented bisemilattices (\OCBSL)\label{tab:laws}.}
    
    \begin{tabular}{rc|rc}
         L9: & $x \sqcup (x \land y) = x$ & 
         L9': & $x \land (x \sqcup y) = x$   \\
         L10:  & $x \sqcup (y \land z) = (x \sqcup y) \land (x \sqcup z)$ & 
         L10': & $x \land (y \sqcup z) = (x \land y) \sqcup (x \land z)$ 
    \end{tabular}
    
    \
    
    \caption{Neither the 
    absorption law L9,L9' nor distributivity L10,L10' hold in \OCBSL. 
    Without L9,L9', the operations $\land$ and $\sqcup$ induce different partial orders.
    If an {\OCBSL} satisfies L10,L10' then it also 
    satisfies L9,L9' and is precisely a Boolean algebra\label{tab:lawsBeyond}.}
\end{table}

\subsection{Problem Statement}

This paper proposes to approximate propositional formula equivalence using a new algorithm that solves exactly the word problem for structures we call orthocomplemented bisemilattices (axiomatized in Table~\ref{tab:laws}),
in only log-linear time.
In general, the word problem for an algebraic theory with signature $S$ and axioms $A$ is the problem of determining, given two terms $t_1$ and $t_2$ in the language of $S$ with free variables, whether $t_1 = t_2$ is a consequence of the axioms. Our main interest in the problem is that orthocomplemented bisemilattices (\OCBSL) are a generalisation of Boolean algebra. This structure satisfies a weaker set of axioms that omits the distributivity law as well as its weaker variant, the absorption law (Table~\ref{tab:lawsBeyond}). Hence, this problem is a relaxation ``up to distributivity'' of the propositional formula equivalence. A positive answer implies formulas are equivalent in all Boolean algebras, hence also in propositional logic.

\begin{definition}[Word Problem for Orthocomplemented Bisemilattices]
Consider the signature with two binary operations $\land, \sqcup$,
unary operation $\neg$ and constants, $0, 1$. 
The {\OCBSL}-word problem is the problem of determining, given two terms 
$t_1$ and $t_2$ in this signature, containing free variables, 
whether $t_1 = t_2$ is a consequence of the axioms L1-L8,L1'-L8' in Table~\ref{tab:laws}.
\end{definition}

\noindent
\textbf{Contribution.}
We present an $\mathcal{O}(n \log^2 (n))$ algorithm for the word problem of orthocomplemented lattices.

We analyze the algorithm to show its correctness and we also present its executable description and a Scala implementation at \url{https://github.com/epfl-lara/OCBSL}.

\subsection{Related Work}

The word problem on \emph{lattices} has been studied in the past. The structure we consider is, in general, \emph{not} a lattice. Whitman \cite{whitmanFreeLattices1941} showed decidability of the word problem on free lattices, essentially by showing that the natural order relation on lattices between two words can be decided by an exhaustive search.
The word problem on \emph{orthocomplemented} \emph{lattices} has been solved typically by defining a suitable sequent calculus for the order relation with a cut rule for transitivity \cite{brunsFreeOrtholattices1976, kalmbachOrthomodularLattices1983}. Because a cut elimination theorem can be proved similarly to the original from Gentzen \cite{Gentzen1935}, the proof space is finite and a proof search procedure can decide validity of the implication in the logic, which translates to the original word problem.

The word problem for free lattices was shown to be in PTIME by Hunt et al. \cite{huntComputationalComplexityAlgebra1987} and the word problem for orthocomplemented lattices was shown to be in PTIME by Meinander \cite{meinanderSolutionUniformWord2010}. Those algorithms essentially rely on similar proof-search methods as the previous ones, but bound the search space. These results make no mention of a specific degree of the polynomial; our analysis suggest that, as described, these algorithms run in $\mathcal O(n^4)$. 
Related techniques of locality have been applied more broadly and also yield
polynomial bounds, with the specific exponents depending on local Horn clauses that axiomatize the theory~\cite{DBLP:journals/jacm/BasinG01, mcallesterAutomaticRecognitionTractability1993}.

Aside from the use in equivalence checking, the problem is additionally of independent interest because {\OCBSL} are a natural weakening of Boolean Algebra and orthocomplemented lattices. They are dual to complemented lattices in the sense illustrated by Figure~\ref{figbsl}. 
A slight weakening of {\OCBSL}, called de Morgan bisemilattice, has been used to simulate electronic circuits \cite{brzozowskiMorganBisemilattices2000, lewisHazardDetectionQuinary1972}. {\OCBSL} may be applicable in this scenario as well. Moreover, our algorithm can also be adapted to decide, in log-linear time, the word problem for this weaker theory.

To the best of our knowledge, no solution was presented in the past for the word problem for orthocomplemented bisemilattices ({\OCBSL}). Moreover, we are not aware of previous log-linear algorithms for the related previously studied theories either. 

\subsection{Overview of the Algorithm}
It is common to represent a term, like a Boolean formula, as an abstract tree. In such a tree, a node corresponds to either a function symbol, a constant symbol or a variable, and the children of a function node represent the arguments of the function. In general, for a symbol function $f$, $f(x,y) \neq f(y,x)$ and the children of a node are stored in a specific order. Commutativity of a function symbol $f$ corresponds to the fact that children of a node labelled by $f$ are instead unordered. Our algorithm thus uses as its starting point a variation of the algorithm of Aho, Hopcroft, and Ullman  \cite{hopcroftDesignAnalysisComputer} for tree isomorphism, as it corresponds to deciding equality of two terms modulo commutativity.
However, the theory we consider contains many more axioms than merely commutativity. Our approach is to find an equivalent set of reduction rules, themselves understood modulo commutativity, that is suitable to compute a normal form of a given formula with respect to those axioms using the ideas of term rewriting \cite{baaderTermRewritingAll1998}. The interest of tree isomorphism in our approach is two-fold: first, it helps to find application cases of our reduction rules, and second, it compares the two terms of our word problem. In the final algorithm, both aspects are realized simultaneously.

\begin{figure}[hbt]
\makebox[\textwidth][c]{
\begin{tabular}{ccc }
    \begin{tikzpicture}[minimum size=10mm, inner sep=0mm, node distance={18mm}, main/.style = {draw, circle}] 
    \node[main] (1) {$a\leq b$}; 
    \node[main] (2) [below left of=1] {$a \sqsubseteq b$};
    \node[main] (3) [below right of=1] {$\neg b \leq \neg a$}; 
    \node[main] (4) [below right of=2] {$\neg b \sqsubseteq \neg a$};
    \draw[double, <->] (1) -- (2);
    \draw[double, <->] (3) -- (4);
    \end{tikzpicture} &
    \begin{tikzpicture}[minimum size=10mm, inner sep=0mm, node distance={18mm}, main/.style = {draw, circle}] 
    \node[main] (1) {$a\leq b$}; 
    \node[main] (2) [below left of=1] {$a \sqsubseteq b$};
    \node[main] (3) [below right of=1] {$\neg b \leq \neg a$}; 
    \node[main] (4) [below right of=2] {$\neg b \sqsubseteq \neg a$};
    \draw[double, <->] (1) -- (4);
    \draw[double, <->] (3) -- (2);
    \end{tikzpicture} &
    \begin{tikzpicture}[minimum size=10mm, inner sep=0mm, node distance={18mm}, main/.style = {draw, circle}] 
    \node[main] (1) {$a\leq b$}; 
    \node[main] (2) [below left of=1] {$a \sqsubseteq b$};
    \node[main] (3) [below right of=1] {$\neg b \leq \neg a$}; 
    \node[main] (4) [below right of=2] {$\neg b \sqsubseteq \neg a$};
    \draw[double, <->] (1) -- (4);
    \draw[double, <->] (3) -- (2);
    \draw[double, <->] (1) -- (2);
    \draw[double, <->] (3) -- (4);
    \draw[double, <->] (1) -- (3);
    \draw[double, <->] (4) -- (2);
    \end{tikzpicture} \\
    && \\
    (a) Complemented lattice & (b) Orthocomplemented bisemilattice & (c) Orthocomplemented lattice
    \end{tabular}
    }
\caption{Bisemilattices satisfying absorption or de Morgan laws.}

\label{figbsl}
\end{figure}
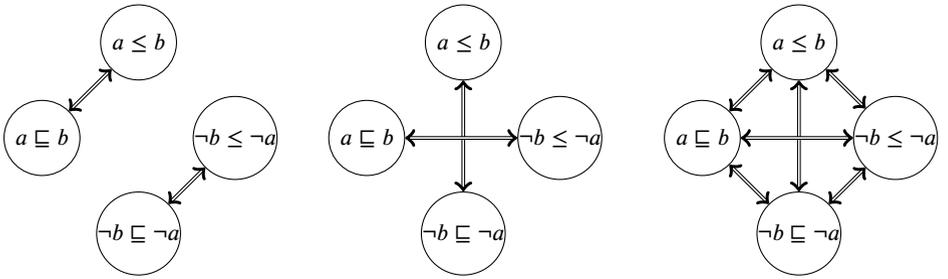

\section{Preliminaries}
\label{sectprel}
\subsection{Lattices and Bisemilattices}
To define and situate our problem, we present a collection of algebraic structures satisfying certain subsets of the laws in Table \ref{tab:laws}.

A structure $(S, \land )$ that is associative (L1), commutative (L2) and idempotent (L3) is a \textbf{semilattice}.
A semilattice induces a partial order relation on $S$ defined by $a \leq b \iff (a\land b) = a$. Indeed, one can verify that $\exists c. (b\land c) = a \iff (b\land a)=a$, from which transitivity follows. Antisymetry is immediate.  In the poset $S$, two elements $a$ and $b$ always have a \textit{greatest lower bound}, or $glb$, $a \land b$. Conversely, a poset such that any two elements have a $glb$ is always a semilattice. 

A structure $(S, \land, 0, 1)$ that satisfies L1, L2, L3, L4, and L5 is a bounded \textbf{upper-semilattice}. Equivalently, $1$ is a maximal element and $0$ a minimal element in the corresponding poset. Similarly, a structure $(S, \sqcup, 0, 1)$ that satisfies L1' to L5' is a bounded \textbf{lower-semilattice}. In that case, we write the corresponding ordering relation $\sqsupseteq$. Note that it points in the direction opposite to $\leq$, so that 1 is always the ``maximum'' element and 0 the ``minimum'' element.

A structure $(S, \land, \sqcup)$ is a \textbf{bisemilattice} if $(S, \land)$ is an upper semilattice and $(S, \sqcup)$ a lower semilattice. There are in general no specific laws relating the two semilattices of a bisemilattice. They can be the same semilattice or completely different. If the bisemilattice satisfies the absorption law (L9), then the two semilattices are related in such a way that $a \leq b \iff a \sqsupseteq b$, i.e. the two orders $\leq$ and $\sqsupseteq$ are equal and the structure is called a lattice.

A bisemilattice is \textbf{consistently bounded} if both semilattices are bounded and if $0_\land = 0_\sqcup = 0$ and $1_\land = 1_\sqcup = 1$. In this paper, we only consider consistently bounded bisemilattices.

A structure $(S, \land, \sqcup, \neg, 0, 1)$ that satisfies L1 to L7 and L1' to L7' is called a \textbf{complemented bisemilattice}, with complement operation $\neg$. A complemented bisemilattice satisfying de Morgan's Law (L8 and L8') is an \textbf{orthocomplemented bisemilattice}.
Similarly, a structure satisfying L1-L9 and L1'-L9' is an \textbf{orthocomplemented lattice}. Both de Morgan laws (L8, L8') and absorption laws (L9 and L9') relate the two semilattices, in a way summarised in Figure~\ref{figbsl}. In bisemilattices, orthocomplementation is (merely) equivalent to $a \leq b \iff \neg b\sqsupseteq \neg a$. 
Indeed, we have:
\[
a \leq b 
\mathop{\iff}^{def} a \land b = a 
\mathop{\iff}^{L8'} \neg a \sqcup \neg b = \neg a  
\mathop{\iff}^{def} \neg b \sqsupseteq \neg a 
\]

Note that, in the presence of L1-L8,L1'-L8', the law of absorption (L9 and L9') is implied by distributivity. In fact, an orthocomplemented bisemilattice with distributivity is a lattice and even a Boolean algebra. In this sense, we can consider orthocomplemented bisemilattices as ``Boolean algebra without distributivity''.

\subsection{Term Rewriting Systems}
We next review basics of term rewriting systems. For a more complete treatment, see \cite{baaderTermRewritingAll1998}.
\begin{definition}
A \textbf{term rewriting system} is a list of rewriting rules of the form $e_l = e_r$ with the meaning that the occurence of $e_l$ in a term $t$ can be replaced by $e_r$. $e_l$ and $e_r$ can contain free variables. To apply the rule,  $e_l$ is unified with a subterm of $t$, and that subterm is replaced by $e_r$ with the same unifier. If applying a rewritting rule to $t_1$ yields $t_2$, we say that $t_1$ reduces to $t_2$ and write $t_1 \rightarrow{t_2}$. We note $\starra$ the transitive closure of $\rightarrow$ and $\starlra$ its transitive symmetric closure.
\end{definition}
Note that an axiomatic system such as L1-L9, L1'-L9' induces a term rewriting system, interpreting equalities from left to right. In that case $t_1 \starlra t_2$ coincides with the validity of the equality $t_1 = t_2$ in the theory given by the axioms.
\begin{definition}
A term rewriting system is terminating if there exists no infinite chain of reducing terms $t_1 \rightarrow t_2 \rightarrow t_3 \rightarrow ...$.
\end{definition}
\begin{fact}
If there is a well-founded order $<$ (or equivalently a measure $m$) on terms such that $t_1 \rightarrow t_2 \implies t_2 < t_1$ (or equivalently $m(t_2)<m(t_1)$) then the term rewriting system is terminating.
\end{fact}

\begin{definition}
A term rewriting system is \textbf{confluent} if 
$$\forall t_1, t_2, t_3. (t_1 \starra t_2 \land t_1 \rightarrow^* t_3) \implies (\exists t_4. t_2 \starra t_4 \land t_3 \starra t_4)$$
\end{definition}
\begin{theorem}[Church-Rosser Property ]\cite[Chapter~2]{baaderTermRewritingAll1998}
A term rewriting system is confluent if and only if 
$\forall t_1, t_2. (t_1 \starlra t_2) \implies (\exists t_3. t_1 \starra t_3 \land t_2\starra t3)$.
\end{theorem}

A terminating and confluent term rewriting system directly provides decidability of the word problem for the underlying structure, as it is poossible to compute the normal form of two terms to check if they are equivalent. Note that commutativity is not a terminating rewriting rule, but similar results holds if we consider the set of all terms, as well as rewrite rules, modulo commutativity \cite[Chapter~11]{baaderTermRewritingAll1998}, \cite{PetersonStickel1981}.
To efficiently manipulate terms modulo commutativity and achieve
log-linear time, we will employ an algorithm
for comparing trees with unordered children.

\section{Directed Acyclic Graph Equivalence}
\label{sect:treeiso}
The structure of formulas with commutative nodes correspond to the usual mathematical definition of a labelled rooted tree, i.e. an acyclic graph with one distinguished vertex (root) where there is no order on the children of a node. For this reason, we use as our starting point the algorithm of Hopcroft, Ullman and Aho for tree isomorphism \cite[Page 84, Example 3.2]{hopcroftDesignAnalysisComputer},
which has also been studied subsequently \cite{LindellTreesLogspace1992,AlogtimeIsomorphism1997}.

This is best described by singly-rooted, labeled, Directed Acyclic Graphs (DAG), which generalize trees. Any directed acyclic graph can be transformed into a tree by duplicating subgraphs corresponding to nodes with multiple parent, as in Figure~\ref{fig:dagtreeexpansion}. This transformation in general results in an exponential blowup in the number of nodes. Conversely, being able to refer to trees with the structure of a DAG can exponentially shrink needed space in some cases.
\begin{figure}[hbt]
  \begin{minipage}[b]{0.45\textwidth}
\center
\begin{tikzpicture}[distance={15mm}, main/.style = {draw, circle}] 
\node[main] (1) {}; 
\node[main] (2) [below left of=1] {};
\node[main] (3) [below right of=1] {}; 
\node[main] (4) [below right of=2] {};
\draw[->] (1) -- (2);
\draw[->] (1) -- (3);
\draw[->] (2) -- (4);
\draw[->] (3) -- (4);
\node[main] (5) [right of=3]{}; 
\node[main] (6) [above right of=5] {};
\node[main] (7) [below right of=6] {}; 
\node[main] (8) [below right of=7] {};
\node[main] (9) [below left of=5] {};
\draw (6) -- (5);
\draw (6) -- (7);
\draw (5) -- (9);
\draw (7) -- (8);

\end{tikzpicture} 
\caption{A DAG and the corresponding Tree}
\label{fig:dagtreeexpansion}
  \end{minipage}
\hfill
  \begin{minipage}[b]{0.45\textwidth}
\center
\begin{tikzpicture}[ main/.style = {draw, circle}]

\node[main] (1) {}; 
\node[main] (2) [below of=1] {};
\node[main] (3) [below left of=2] {}; 
\node[main] (4) [below right of=2] {};
\draw[->] (1) -- (2);
\draw[->] (2) -- (3);
\draw[->] (2) -- (4);

\node[main] (5) [right of=4]{}; 
\node[main] (6) [above  of=5] {};
\node[main] (7) [above  of=6] {}; 
\draw[->] (7) -- (6);
\draw[->] (6) edge[bend left=40] (5);
\draw[->] (6) edge[bend right=40] (5);

\end{tikzpicture} 
\caption{Two equivalent DAGs with different number of nodes.}
\label{fig:dagequiv}
\end{minipage}
 \end{figure}

Checking for equality between \emph{ordered} trees or DAGs is easy in linear time: we simply recursively check equality between the children of two nodes.
\begin{definition}
\label{def:dagequiv1}Two ordered nodes $\tau$ and $\pi$ with children $\tau_0,...,\tau_m$ and $\pi_0,...,\pi_n$ are equivalent (noted $\tau \sim \pi$) iff
$$
label(\tau) = label(\pi)\text{, }
m=n\text{ and }
\forall i<n, \tau_i \sim \pi_i
$$
\end{definition}
For unordered trees or DAG, the equivalence checking is less trivial, as the naive algorithm has exponential complexity due to the need to find the adequate permutation.
\begin{definition}
\label{def:dagequiv2}Two unordered nodes  $\tau$ and $\pi$ with children $\tau_0,...,\tau_m$ and $\pi_0,...,\pi_n$ are equivalent (noted $\tau \sim \pi$) iff
$$
label(\tau) = label(\pi)\text{, }
m=n \text{ and }
\text{there exists a permutation $p$ s.t. } \forall i<n, \tau_{p(i)} \sim \pi_i
$$
\end{definition}

For trees, note that this definition of equivalence corresponds exactly to isomorphism.
It is known that DAG-isomorphism is GI-complete, so it is conjectured to have complexity greater than PTIME. Fortunately, this does not prevent our solution because our notion of equivalence on DAGs is not the same as isomorphism on DAGs. In particular, two DAGs can be equivalent without having the same number of nodes, i.e. without being isomorphic, as Figure \ref{fig:dagequiv} illustrates.

\begin{algorithm}[hbt]
\SetAlgoLined
\SetKwInOut{Input}{input}
    \SetKwInOut{Output}{output}
    \Input{two DAGs $\tau$ and $\pi$}
    \Output{True if $\tau$ and $\pi$ are equivalent, False else.}
  $codes \gets $HashMap[(String, List[Int]), Int]\;  
  $map \gets $HashMap[Node, Int]\;
 $s_\tau:List \gets ReverseTopologicalOrder(\tau)$\;
 $s_\pi:List \gets ReverseTopologicalOrder(\pi)$\;
  \For{($n$:Node in $s_\tau \mbox{\tt ++} s_\pi$)}{
   $l_n \gets [map(c)$ for $c$ in $children(n)]$\;
   $r_n \gets (label(n), sort(l_n))$\;
   \eIf{$codes$ contains $r$}
   {$map(n)\gets codes(r_n)$\;}
   {$codes(r_n)\gets codes.size$\;
   $map(n)\gets codes(r_n)$\;}
 }
 $\textbf{return}$ $map(\tau) == map(\pi)$
 \caption{DAG equivalence. The operator \mbox{\tt ++} is list concatenation.}
 \label{algo:dag}
\end{algorithm}

Algorithm~\ref{algo:dag} is the generalization of Hopcroft, Ullman and Aho's algorithm. It decides in log-linear time if two labelled (unordered) DAGs are equivalent according to definition~\ref{def:dagequiv1}.

The algorithm works bottom to top. We first sort the DAG in reverse topological order using, for example, Kahn's algorithm\cite{kahnTopologicalSortingLarge1962}. This way, we explore the DAG starting from a leaf and finishing with the root. It is guaranteed that when we treat a node, all its children have already been treated.

The algorithm recursively assigns codes to the nodes of both DAGs recursively. In the unlabelled case:
\begin{list}{•}{}
\item The first node, necessarily a leaf, is assigned the integer $0$
\item The second node gets assigned $0$ if it is a leaf or $1$ if it is a parent of the first node
\item For any node, the algorithm makes a list of the integer assigned to that node's children and sort it (if the node is commutative). We call this the signature of the node. Then it checks if that list has already been seen. If yes, it assigns to the node the number that has been given to other nodes with the same signature. Otherwise, it assigns a new integer to that node and its signature.
\end{list}

\begin{lemma}[Algorithm~\ref{algo:dag} Correctness]
The codes assigned to any two nodes $n$ and $m$ of $s_\tau \mbox{\tt ++} s_\pi$ are equal if and only $n \sim m$.
\end{lemma}
\begin{proof}

Let $n$ and $m$ denote any two nodes in what follows. By induction on the height of $n$:
\begin{itemize}
\item In the case where $n$ is a leaf, we have $r_n = (label(n), Nil)$. Note that for any node $n$, $map(n) = codes(r_n)$. Since every time the map $codes$ is updated, it is with a completely new number, $codes(r_n) = codes(r_m)$ if and only if $r_n = r_m$, i.e. iff $label(m) = label(n)$ and $m$ has no children (like $n$).
\item In the case where $n$ has children $n_i$, again $codes(r_n) = codes(r_m)$ if and only if $r_m = r_n$, which is equivalent to ($label(m)=label(n)$ and $sort(l_m) = sort(l_n)$. This means this means there is a permutation of children of $n$ such that $\forall i, codes(n_{p(i)}) = codes(m_i)$. By induction hypothesis, this is equivalent to $\forall i, n_{p(i)} \sim m_i$. Hence we find that $map(n) = map(m)$ if and only if both:
\begin{enumerate}
\item Their labels are equal
\item There exist a permutation $p$ s.t. $n_{p(i)} \sim m_i$
\end{enumerate}
i.e $n$ and $m$ have the same code if and only if $n \sim m$.
\end{itemize}
\end{proof}
\begin{corollary}
The algorithm returns $True$ if and only if $\tau \sim \pi$
\end{corollary}
\paragraph{Time Complexity.}
Using Kahn's algorithm, sorting $\tau$ and $\pi$ is done in linear time. Then the loop touches every node a single time. Inside the loop, the first line takes linear time with respect to the number of children of the node and the second line takes log-linear time with respect to the number of children. Since we use HashMaps, the last instructions take effectively constant time (because hash code is computed from the address of the node and not its content).

So for general DAG, the algorithm runs in time at most log-quadratic in the number of nodes. Note however that for DAGs with bounded number of children per node as well as for DAGs with bounded number of parents per nodes, the algorithm is log-linear. In fact, the algorithm is log-linear with respect to the total number of edges in the graph. For this reason, the algorithm is still only log-linear in input size. It also follows that the algorithm is always at most log-linear with respect to the tree or formula underlying the DAG, which may be much larger than the DAG itself. Moreover, there exists cases where the algorithm is log-linear in the number of nodes, but the underlying tree is exponentially larger. The full binary symmetric graph is such an example.

\section{Word Problem on Orthocomplemented Bisemilattices}
\label{sect:wordprob}
We will use the previous algorithm for DAG equivalence applied to a formula in the language of bisemilattices $(S, \land, \sqcup)$ to account for commutativity (axioms L1, L1'), but we need to combine it with the remaining axioms. From now on we work with axioms L1-L8, L1'-L8' in Table~\ref{tab:laws}. The plan is to express those axioms as reduction rules. Of rules L2-L8 and L2'-L8', all but L8 and L8' reduce the size of the term when applied from left to right, and hence seem suitable as rewrite rules. 

It may seem that the simplest way to deal with de Morgan law is to use it (along with double negation elimination) to transform all terms into negation normal form. It happens, however, that doing this causes troubles when trying to detect application cases of rule L7 (complementation). Indeed, consider the following term:
$$
f = (a \land b) \sqcup \neg(a \land b)
$$
Using complementation it clearly reduces to $1$, but pushing into negation-normal form, it would first be transformed to $(a\land b) \sqcup (\neg a \lor \neg b)$.  To detect that these two disjuncts are actually opposite requires to recursivly verify that $\neg (a\land b) = (\neg a \lor \neg b)$.

It is actually simpler to apply de Morgan law  the following way:
$$
x\land y = \neg(\neg x \sqcup \neg y)
$$
Instead of removing negations from the formula, we remove one of the binary semilattice operators. (Which one we keep is arbitrary; we chose to keep $\sqcup$.) Now, when we look if rule L7 can be applied to a disjunction node (i.e. two children $y$ and $z$ such that $y = \neg z$), there are two cases: if $x$ is not itself a negation, i.e. it starts with $\sqcup$, we compute $\neg x$ code from the code of $x$ in constant time. If $x = \neg x'$ then $\neg x \sim x'$ so the code of $\neg x$ is simply the code of $x$, in constant time as well. Hence we obtain the code of all children and their negation and we can sort those codes to look for collisions, all of it in time linear in the number of children.

We now restate the axioms L1-L8 ,L1'-L8' in this updated language in Table~\ref{tab:lawsA}.
\begin{table}[hbt]
    \[\begin{array}{rl|rl}
         A1: & \bigsqcup(..., x_i,x_j,...) = \bigsqcup(..., x_j,x_i,...)  & A1': & \neg\bigsqcup(\neg x, \neg y) = \neg \bigsqcup(\neg y, \neg x)\\
         A2: & \bigsqcup(\vec x, \bigsqcup( \vec y)) = \bigsqcup(\vec x, \vec y)  & A2': & \neg\bigsqcup(\neg\vec x, \neg\neg\bigsqcup( \neg\vec y)) = \neg\bigsqcup(\neg\vec x, \neg\vec y)\\
         
            &   \bigsqcup(x)=x & & \\
         A3: & \bigsqcup(x,x, \vec y) = \bigsqcup(x, \vec y)  & A3': & \neg \bigsqcup(\neg x, \neg x, \neg \vec y) = \neg\bigsqcup(\neg x, \neg \vec y) \\
         A4: & \bigsqcup(1, \vec x)= 1  & A4': & \neg \bigsqcup(\neg 0, \neg \vec y)= 0 \\
         A5: & \bigsqcup(0, \vec x) = \bigsqcup(\vec x)  & A5': & \neg\bigsqcup(\neg 1, \neg\vec x) = \neg\bigsqcup(\neg \vec x)\\
         A6: & \neg \neg x = x  &  &  \\
         A7: & \bigsqcup(x,\neg x, \vec y) = 1  & A7': & \neg \bigsqcup(\neg x,\neg\neg x, \neg\vec y) = 0 \\
         A8: & \neg\bigsqcup(x_1,...x_i) = \neg\bigsqcup(\neg\neg x_1,...\neg\neg x_i)  & A8': &  \neg\neg\bigsqcup(\neg x_1,...\neg x_i) = \bigsqcup(\neg x_1,...\neg x_i) \\
    \end{array}\]    
    \caption{Laws of algebraic structures $(S, \sqcup, 0, 1, \neg)$, equivalent to L1-L8, L1-L8' under de Morgan transformation.}
    \label{tab:lawsA}
\end{table}

It is straightforward and not surprising that axiom A8 as well as A1'-A8' all follow from axioms A1-A7, so A1-A7 are actually complete for our theory. 

\subsection{Confluence of the Rewriting System}

In our equivalence algorithm, A1 is taken care of by the arbitrary but consistent ordering of the nodes. Axioms A2-A7 form a term rewriting system. Since all those rules reduce the size of the term, the system is terminating in a number of step linear in the size of the term. We will next show that it is confluent. We will thus obtain the existence of a normal form for every term, and will finally show how our algorithm computes that normal form.

\begin{definition}
Consider a pair of reduction rules $l_0\rightarrow r_0$ and $l_1\rightarrow r_1$ with disjoint sets of free variables such that $l_0 $ = $D[s]$, $s$ is not a variable and $\sigma$ is the most general unifier of $\sigma s = \sigma l_1$. $(\sigma r_0, (\sigma D)[\sigma r_2])$ is called a critical pair.
\end{definition}
Informally, a critical pair is a most general pair of term (with respect to unification) $(t_1, t_2)$ such that for some $t_0$, $t_0 \rightarrow t_1$ and $t_0 \rightarrow t_2$ via two ``overlapping'' rules. They are found by matching the left-hand side of a rule with a non-variable subterm of the same or another rule. 
\begin{example}[Critical Pairs]\leavevmode
\begin{enumerate}
\item Matching left-hand side of A6 with the subterm $\neg x$ of rule A7, we obtain the pair
$$
(1, \bigsqcup(\neg x, x, \vec y))
$$
which arises from reducing the term $t_0 = \bigsqcup(\neg x,\neg \neg x, \vec y)$ in two different ways.
\item Matching left-hand sides of A2 and A7 gives
$$
(\bigsqcup(\vec x, \vec y, \neg \bigsqcup(\vec y)), 1)
$$
which arise from reducing $\bigsqcup(\vec x, \bigsqcup(\vec y), \neg \bigsqcup(\vec y))$
\item Matching left-hand sides of A5 and A7 gives
$$
(\neg 0, 1)
$$
which arise from reducing $0 \sqcup \neg 0$ in two different ways.
\end{enumerate}
\end{example}
\begin{proposition}[{\cite[Chapter~6]{baaderTermRewritingAll1998}}]
\label{critpairslemma}
A terminating term rewriting system is confluent if and only if all critical pairs are joinable.
\end{proposition}
In the first of the previous examples, the pair is clearly joinable by commutativity and a single application of rule A7 itself. The second example is more interesting. Observe that $\bigsqcup(\vec x, \vec y, \neg \bigsqcup(\vec y)) = 1$ is a consequence of our axiom, but the left part cannot be reduced to 1 in general in our system. To solve this problem we need to add the rule A9: $\bigsqcup(\vec x, \vec y, \neg \bigsqcup(\vec y)) = 1$. Similarly, the third example forces us to add A10: $\neg 0 = 1$ to our set of rules. From A10 and A6 we then find the expected critical pair A11: $\neg 1 = 0$.

\begin{table}[hbt]
    \centering 
    \[\begin{array}{rl}
         A1: & \bigsqcup(..., x_i,x_j,...) = \bigsqcup(..., x_j,x_i,...)  \\
         A2: & \bigsqcup(\vec x, \bigsqcup( \vec y)) = \bigsqcup(\vec x, \vec y)  \\
            & \bigsqcup(x)=x \\
         A3: & \bigsqcup(x,x, \vec y) = \bigsqcup(x, \vec y)  \\
         A4: & \bigsqcup(1, \vec x)= 1  \\
         A5: & \bigsqcup(0, \vec x) = \bigsqcup(\vec x)  \\
         A6: & \neg \neg x = x \\
         A7: & \bigsqcup(x,\neg x, \vec y) = 1   \\
         A9: & \bigsqcup(\vec x, \vec y, \neg \bigsqcup(\vec y)) = 1  \\
         A10: & \neg 0 = 1  \\
         A11: & \neg 1 = 0   \\
    \end{array}\]
    \caption{Terminating and confluent set of rewrite rules equivalent to L1-L8, L1'-L8'
    \label{tab:truelawsA}}
\end{table}

\subsection{Complete Terminating Confluent Rewrite System}
The analysis of all possible pairs of rules to find all critical pairs is straightforward. It turns out that the A9, A10 and A11 are the only rules we need to add to our system to obtain confluence. The complete list of critical pairs for rules A2-A11 has been checked but is omitted for matter of space. All those pairs are joinable, i.e. reduce to the same term, which implies, by Proposition~\ref{critpairslemma}, that the system is confluent. Table~\ref{tab:truelawsA} shows the complete set of reduction rules (as well as commutativity).

Since the system A2-A11 considered over the language $(S, \bigsqcup, \neg, 0, 1)$ modulo commutativity of $\bigsqcup$ is terminating and confluent, it implies the existence of a normal form reduction. For any term $t$, we note its normal form $t\downarrow$. In particular, for any two terms $t_1$ and $t_2$, $t_1 = t_2$ in our theory if and only if $t_1\starlra t_2$. We finally reach our conclusion: An algorithm that computes the normal form (modulo commutativity) of any term gives a decidable procedure for the word problem for orthocomplemented bisemilattices

\section{Algorithm and Complexity}
\label{sect:algo}
The rewriting system readily gives us a quadratic algorithm. Indeed, using our base algorithm for DAG equivalence, we can check for application cases of any one of our rewritting rule A2-A11 in linear time. Since a term can only be reduced up to $n$ times, the total time spent before finding the normal form of a term is at most quadratic. It is however possible to find the normal form of a term in a single pass of our equivalence algorithm, resulting in a more efficient algorithm.

\subsection{Combining Rewrite Rules and Tree Isomorphism}
We give an overview on how to combine rules A2-A7, A9, A10, A11 within the tree isomorphism algorithm, which we present using Scala-like~\footnote{\url{https://www.scala-lang.org/}} pseudo code in Figure~\ref{algo:fullalgorithm}. For conciseness, we omit the dynamic programming optimizations allowed  by structure sharing in DAGs (which would store the normal form and additionally check if a node was already processed.) For each rule, we indicate the most relevant lines of the algorithm  in Figure~\ref{algo:fullalgorithm}. 
\RestyleAlgo{tworuled}
\LinesNotNumbered

\paragraph{A2} (Associativity, Lines 10, 20, 32, 42) When analysing a $\bigsqcup$ node, after the recursive call, find all children that are $\bigsqcup$ themselves and replace them by their own children. 
This is simple enough to implement but there is actually a caveat with this in term of complexity. We will come back to it in section \ref{sect:algo}.

\paragraph{A3} (Idempotence, Lines 8, 31, 35 ) This corresponds to the fact that we eliminate duplicate children in disjunctions. When reaching a $\bigsqcup$ node, after having sorted the code of its children, remove all duplicates before computing its own code.

\paragraph{A4, A5} (Bounds, Lines 8, 31, 35, 11, 36)  To account for those axioms, we reserve a special code for the nodes $1$ and $0$. For A4, when we reach some $\bigsqcup$ node, if it has $1$ as one of its children, we accordingly replace the whole node by $1$. For A5, we just remove nodes with the same codes as $0$ from the parent node before computing its own code.

\paragraph{A6} (Involution, Lines 17, 22) When reaching a negation node, if its child is itself a negation node, replace the parent node by its grandchildren before assigning it a code.

\paragraph{A7} (Complement, Lines 11, 36) As explained earlier, our representation of nodes let us do the following to detect cases of A7: First remember that we already applied double negation elimination, so that two ``opposite'' nodes cannot both start with a negation. Then we can simply separate the children between negated and non-negated (after the recursive call), sort them using their assigned code and look for collisions.

\paragraph{A9} (Also Complement, Lines 11, 36) This rule is slightly more tricky to apply. When analysing a $\bigsqcup$ node $x$, after computing the code of all children of $x$, find all children of the form $\neg \bigsqcup$. For every such node, take the set of its own children and verify if it is a subset of the set of all children of $x$. If yes, then rule A9 applies. Said otherwise, we look collisions between grandchildren (through a negation) and children of every $\bigsqcup$ node.

\paragraph{A10, A11} (Identities, Lines 17, 26) These rules are simple. In a $\neg$ node, if its child has the same code as $0$ (resp $1$), assign code $1$ (resp $0$) to the negated node.

\subsection{Case of Quadratic Runtime for the Basic Algorithm}

All the rules we introduced in the previous section into Algorithm~$\ref{algo:dag}$ take time (log)linear in the number of children of a node to apply, which is not more than the time we spent in the DAG/tree isomorphism algorithm. For A3, Checking for duplicates is done in linear time in an ordered data structure. A4 and A5 (Bounds) consist in searching for specific values, which take logarithmic time in the size of the list. A6 (Involution) takes constant time. A7 (Complement) Is detected by finding a collision between two separate ordered lists, also easily done in (log) linear time. A9 (Also complement) consists in verifying if grandchildren of a node are also children, and since children are sorted this takes log-linear time in the number of grandchildren. Since a node is the grandchild of only one other node, the same computation as in the original algorithm holds. A10 and A11 take constant time. Hence, the total time complexity is $\mathcal{O}(n\log(n))$, as in the algorithm for tree isomorphism.

As stated previously, the time complexity analysis crucially rely on the fact that in a tree, a node is never the child (or grandchild) of more than one node. However, this is generally not true in the presence of associativity. Indeed consider the term represented in Figure \ref{fig:termassociativecounterexample}.
\begin{figure}[hbt]
\center
\begin{tikzpicture}[node distance={15mm}, main/.style = {draw, circle}] 
\node[main] (1) {$\bigsqcup$}; 
\node[main] (2) [right of=1] {$\bigsqcup$};
\node[main] (3) [right of=2] {$\bigsqcup$}; 
\node[main] (4) [right of=3] {$\bigsqcup$};
\node[main] (5) [right of=4] {$\bigsqcup$};
\draw[->] (1) -- (2);
\draw[->] (2) -- (3);
\draw[->] (3) -- (4);
\draw[->] (4) -- (5);

\node[main] (6) [below right of=1] {$x_1$};
\node[main] (7) [below right of=2] {$x_2$};
\node[main] (8) [below right of=3] {$x_3$};
\node[main] (9) [below right of=4] {$x_4$};
\node[main] (10) [above right of=5] {$x_5$};
\node[main] (11) [below right of=5] {$x_6$};
\draw[->] (1) -- (6);
\draw[->] (2) -- (7);
\draw[->] (3) -- (8);
\draw[->] (4) -- (9);
\draw[->] (5) -- (10);
\draw[->] (5) -- (11);
\end{tikzpicture} 
\caption{A term with quadratic runtime}
\label{fig:termassociativecounterexample}
\end{figure}
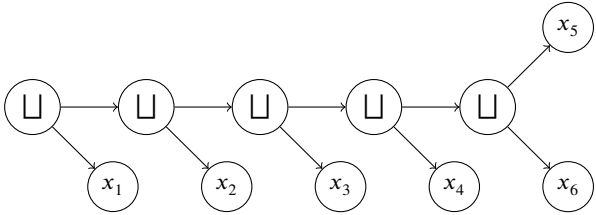
The 5th $\bigsqcup$ has 2 children, but after applying A2, the 4th has 3 children, the 3rd has 4 children and so on. On the generalization of such an example, since an $x_i$ is the child of all higher $\bigsqcup$, our key property does not hold and the algorithm runtime would be quadratic. It's not hard to verify Of course, such a simple counterexample is easily solved by applying a leading pass of associativity reduction before actually running the whole algorithm. It turns out however that it is not sufficient, since cases of associativity can appear after the application of the other A-rules.

In fact, there is only one rule that can creates case of rule A2, and this rule is A6 (Involution). The remaining rules whose right-hand side can start with a $\bigsqcup$ have their left-hand side already starting with $\bigsqcup$. It may seem simple enough to also apply double negation elimination in a leading pass, but unfortunately, cases of A6 can also be created from other rules. It is easy to see, for similar reasons, that only the application of A2b ($\bigsqcup(x) = x$) can create such cases. And unfortunately, such cases of A2b can arise from rules A3 and A5 which can only be detected using the full algorithm. To summarize, the typical problematic case is depicted in Figure~\ref{fig:termassociativecounterexample2}. This term is clearly equivalent to $\bigsqcup(x_1, x_2, x_3, x_4)$, but to detect it we must first find that $z_1$ and $z_2$ are equivalent to $0$, so we cannot simply solve it with an early pass.

\begin{figure}[hbt]
\center
\begin{tikzpicture}[node distance={13mm}, main/.style = {draw, circle}] 
\node[main] (1) {$\bigsqcup$}; 
\node[main] (2) [right of=1] {$\neg$};
\node[main] (3) [right of=2] {$\bigsqcup$}; 
\node[main] (4) [right of=3] {$\neg$};
\node[main] (5) [right of=4] {$\bigsqcup$};
\node[main] (6) [right of=5] {$\neg$};
\node[main] (7) [right of=6] {$\bigsqcup$}; 
\node[main] (8) [right of=7] {$\neg$};
\node[main] (9) [right of=8] {$\bigsqcup$};
\draw[->] (1) -- (2);
\draw[->] (2) -- (3);
\draw[->] (3) -- (4);
\draw[->] (4) -- (5);
\draw[->] (5) -- (6);
\draw[->] (6) -- (7);
\draw[->] (7) -- (8);
\draw[->] (8) -- (9);

\node[main] (10) [below right of=1] {$x_1$};
\node[main] (11) [below right of=3] {$z_1{\sim} 0$};
\node[main] (12) [below right of=5] {$x_2$};
\node[main] (13) [below right of=7] {$z_2{\sim} 0$};
\node[main] (14) [above right of=9] {$x_3$};
\node[main] (15) [below right of=9] {$x_4$};
\draw[->] (1) -- (10);
\draw[->] (3) -- (11);
\draw[->] (5) -- (12);
\draw[->] (7) -- (13);
\draw[->] (9) -- (14);
\draw[->] (9) -- (15);
\end{tikzpicture} 
\caption{A non-trivial term with quadratic runtime}
\label{fig:termassociativecounterexample2}
\end{figure}
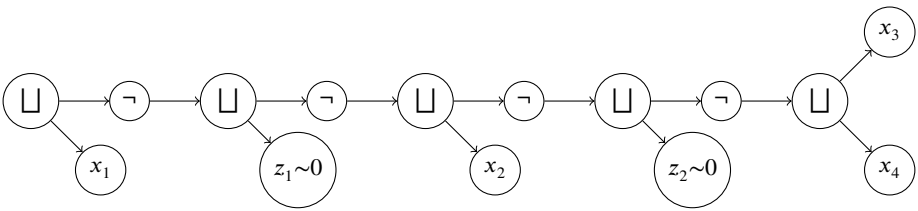

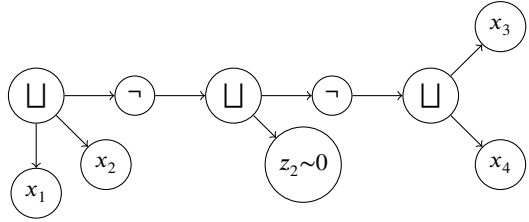
\begin{figure}[hbt]
\center
\begin{tikzpicture}[node distance={13mm}, main/.style = {draw, circle}] 
\node[main] (1) {$\bigsqcup$}; 
\node[main] (6) [right of=1] {$\neg$};
\node[main] (7) [right of=6] {$\bigsqcup$}; 
\node[main] (8) [right of=7] {$\neg$};
\node[main] (9) [right of=8] {$\bigsqcup$};

\draw[->] (1) -- (6);
\draw[->] (6) -- (7);
\draw[->] (7) -- (8);
\draw[->] (8) -- (9);

\node[main] (10) [below  of=1] {$x_1$};
\node[main] (12) [ below right of=1] {$x_2$};
\node[main] (13) [below right of=7] {$z_2{\sim} 0$};
\node[main] (14) [above right of=9] {$x_3$};
\node[main] (15) [below right of=9] {$x_4$};
\draw[->] (1) -- (10);
\draw[->] (1) -- (12);
\draw[->] (7) -- (13);
\draw[->] (9) -- (14);
\draw[->] (9) -- (15);
\end{tikzpicture} 
\caption{the term of Figure~\ref{fig:termassociativecounterexample2} during the algorithm's execution}
\label{fig:termafteronereduction}
\end{figure}

\subsection{Log-Linear Time}
Fortunately, we can solve this problem at a logarithmic-only price.
Observe that if we are able to detect early nodes which would cancel to $0$, the problem would not exist: When analysing a node, we would first call the algorithm on all subnodes equivalent to 0, remove them and then when there is a single children left, remove the trivial disjunct, the double negation and the successive disjunction (as in Figure~\ref{fig:termassociativecounterexample2}) before doing the recursive call on the unique nontrivial child. However, we of course cannot know in advance which child will be equivalent to $0$.

Moreover note (still using Figure~\ref{fig:termassociativecounterexample2}) that if the $z$-child is as large as the non-trivial node, then even if we do the ``useless'' work, we at least obtain that the size a tree is divided by two, and hence the potential depth of the tree as well. By standard complexity analysis, the time penalty would only be a logarithmic factor. This suggests the following solution.

When analysing a node, make the recursive calls on children in order of size, starting with the smallest up to the second biggest. If any of those children are non-zero, proceed as normal. If all (but possibly the last) children are equivalent to zero, then replace the current node by its biggest (and at this point non-analyzed) child, i.e. apply second half of rule A2 (associativity). If applicable, apply double negation elimination and associativity as well before continuing the recursive call.

We illustrate this on the example of Figure~\ref{fig:termassociativecounterexample2}. Consider the algorithm when reaching the second $\bigsqcup$ node. There are two cases:
\begin{enumerate}
    \item Suppose $z_1$ is a smaller tree than the non-trivial child. In this case the algorithm will compute a code for $z_1$, find that it is 0 and delete it. Then the non trivial node is a single child so the whole disjuntion is removed. Hence, the double negation can be removed and the two consecutive disjunction of $x_1$ and $x_2$ merged, obtaining the term illustrated in Figure \ref{fig:termafteronereduction}. In particular we did not compute a code for the two deleted $\bigsqcup$ nodes, which is exactly what we wanted for our initial analysis.
    \item Suppose $z_1$ is larger tree than the non-trivial child. In this case, we would first recursively compute the code of the non-trivial child and then detect that $z_1\sim 0$. We indeed computed the code of the disjunction that contains $x_2$ when it was unnecessary since we apply associativity anyway. This ``useless'' work consists in sorting and applying axioms to the true children of the node (in this case $x_2, x_3$ and $x_4$) and takes time quasilinear in the number of such children. In particular, it is bounded by the size of the subtree itself and we know it is the smallest of the two. 
\end{enumerate}

\begin{figure}[p]
\begin{lstlisting}[language=scala,numbers=left]
def equivalentTrees(tau: Term, pi: Term): Boolean =
    val codesSig: HashMap[(String, List[Int]), Int] = Empty
    codesSig.update(("zero", Nil), 0); codesSig.update(("one", Nil), 1)
    val codesNodes: HashMap[Term, Int] = Empty
    def updateCodes(sig: (String, List[Int]), n: Node): Unit = ...// codesSig, codesNodes
    def bool2const(b:Boolean): String = if b then "one" else "zero"
    def rootCode(n: Term): Int =
        val L = pDisj(n, Nil).map(codesNodes).sorted.filter(_$\neq$ 0).distinct
        if L.isEmpty          then ("zero", Nil), n)
        else if L.length == 1 then codesNodes.update(n, L.head)
        else if L.contains(1) or checkForContradiction(L) then updateCodes(("one", Nil), n)
        else updateCodes(("or", L), n)
        codesNodes(n)
    def pDisj(n:Node, acc:List[Node]): List[Node] = n match
        case Variable(id) $\RA$ updateCodes((id.toString, Nil), n); return n :: acc
        case Literal(b) $\RA$ updateCodes((bool2const(b), Nil), n); return n :: acc
        case Negation(child) $\RA$ pNeg(child, n, acc)
        case Disjunction(children) $\RA$ children.foldleft(acc)(pDisj)    
    def pNeg(n:Node, parent:Node, acc:List[Node]): List[Node] = n match // under negation
        case Negation(child) $\RA$ pDisj(child, acc)
        case Variable(id) $\RA$ updateCodes((id.toString, Nil), n)
                             updateCodes(("neg", List(codesNodes(n))), parent)
                             List(parent)::acc
        case Literal(b) $\RA$ updateCodes((bool2const(b), Nil), n)
                           updateCodes((bool2const(!b), Nil), parent)
                           List(parent)::acc
        case Disjunction(children) $\RA$
            val r0 = orderBySize(children)
            val r1 = r0.tail.foldLeft(Nil)(pDisj) 
            val r2 = r1.map(CodesNodes).sorted.filter(_$\neq$ 0).distinct
            if isEmpty(r2) then pNeg(r0.head, parent, acc)
            else val s1 = pDisj(r0.head, r1)
                 val s2 = s1 zip (s1 map CodesNodes)
                 val s3 = s2.sorted.filter(_$\neq$ 0).distinct //all wrt 2nd element
                 if s3.contains(1) or checkForContradiction(s3)
                 then updateCodes(("one", Nil), n); updateCodes(("zero", Nil), parent)
                      List(parent)::acc                
                 else if isEmpty(s3) then updateCodes(("zero", Nil), n)
                                          updateCodes(("one", Nil), parent)
                                          List(parent)::acc
                 else if s3.length == 1 then pNeg(s3.head._1, parent, acc)
                 else updateCodes(("or", s3 map (_._2)), n)
                      updateCodes(("neg", List(n)), parent)
                      List(parent)::acc
    return rootCode(tau) == rootCode(pi)
\end{lstlisting}
\caption{Final Algorithm. We assume \lstinline{distinct} runs in log-linear time on sorted lists. The \lstinline{checkForContradiction} method detects application cases of A7 and A9 (Complement) }

\label{algo:fullalgorithm}
\end{figure}

Note that the exact same situation can arrise from the use of rule A3 (idempotence), but here trivialy the two subtrees must have the same number of (real) subnodes, so that the same reasonning holds.

We compute the penalty of useless work we incur by computing children of a node $n$ in the wrong order, i.e. by computing a non-0 child $n_w$ when all other are 0. $n_w$ cannot be the largest child of $n$ for otherwise we would have found that all other children are $0$ before needing to compute $n_w$. Hence 
$
|n_w| \leq |n|/2
$.
It follows that the total amount of useless work is bounded by $\log(|n|) \cdot W(n)$, where
$$
W(n) \leq |n|/2 + \sum_i W(n_i)
$$
Where $\sum_i |n_i| < |n|$
It is clear that $W(n)$ is maxed when $n$ has exactly two children of equal size:
$$
W(n) \leq |n|/2 + 2\cdot W(n/2)
$$
Finally, by the Master Theorem for complexity or simply by observing that we can divide $n$ by $2$ only $\log(n)$ times and hence
$$
W(n) \leq \sum_{m=1}^{\log(n)} 2^m\cdot |n|/2^m
$$
we obtain $W(n) = \mathcal{O}(|n|\log(|n|))$ and hence the total runtime is $\mathcal{O}\left(n(\log n)^{2}\right)$.

\section{Conclusion}
We have described a decision procedure with log-linear time complexity for the word problem on orthocomplemented bisemilattices.
This algorithm can also be simplified to apply to weaker theories. In the other direction, we believe it can be reinforced to decide some stronger theories (still weaker than Boolean algebras) in polynomial time.
While the word problem for orthocomplemented \emph{lattices} was known to be in PTIME \cite{huntComputationalComplexityAlgebra1987} and as such the membership of orthocomplemented \emph{bisemilattices} in PTIME may not come as a surprise, this is, to the best of our knowledge, the first time that this result was explicitly been stated, and the first time that an algorithm with such low worst-case complexity was proposed for a similar problem.

Our algorithm is efficient and, according to our preliminary experience, easy to implement in practice. It can be used as an approximation for Boolean algebra equivalence, and we plan to use it as the basis of a kernel for a proof assistant. We also envision possible uses of the algorithm in SMT and SAT solvers. Our algorithm notably does not perform any form of instantiation of variables, which is one reason why it can deal with associativity and commutativity efficiently.  It is able to detect many natural and non-trivial cases of equivalence even on formulas that may be too large for existing solvers to deal with, so it may also complement an existing repertoire of subroutines used in more complex reasoning tasks. For a minimal working implementation in Scala closely following Figure~\ref{algo:fullalgorithm}, see \url{https://github.com/epfl-lara/OCBSL}.

{\raggedcolumns
\bibliographystyle{splncs04}
\bibliography{biblio.bib,vkuncak.bib}

\begin{thebibliography}{10}
\providecommand{\url}[1]{\texttt{#1}}
\providecommand{\urlprefix}{URL }
\providecommand{\doi}[1]{https://doi.org/#1}

\bibitem{baaderTermRewritingAll1998}
Baader, F., Nipkow, T.: Term {{Rewriting}} and {{All That}}. {Cambridge
  University Press}, {Cambridge} (1998). \doi{10.1017/CBO9781139172752}

\bibitem{barrettCVC42011}
Barrett, C., Conway, C.L., Deters, M., Hadarean, L., Jovanovi{\'c}, D., King,
  T., Reynolds, A., Tinelli, C.: {{CVC4}}. In: Gopalakrishnan, G., Qadeer, S.
  (eds.) Computer {{Aided Verification}}. pp. 171--177. Lecture {{Notes}} in
  {{Computer Science}}, {Springer}, {Berlin, Heidelberg} (2011).
  \doi{10.1007/978-3-642-22110-1_14}

\bibitem{DBLP:journals/jacm/BasinG01}
Basin, D.A., Ganzinger, H.: Automated complexity analysis based on ordered
  resolution. J. {ACM}  \textbf{48}(1),  70--109 (2001).
  \doi{10.1145/363647.363681}

\bibitem{brunsFreeOrtholattices1976}
Bruns, G.: Free {{Ortholattices}}. Canadian Journal of Mathematics
  \textbf{28}(5),  977--985 (Oct 1976). \doi{10.4153/CJM-1976-095-6}

\bibitem{bruttomessoOpenSMTSolver2010a}
Bruttomesso, R., Pek, E., Sharygina, N., Tsitovich, A.: The {{OpenSMT Solver}}.
  In: Hutchison, D., Kanade, T., Kittler, J., Kleinberg, J.M., Mattern, F.,
  Mitchell, J.C., Naor, M., Nierstrasz, O., Pandu~Rangan, C., Steffen, B.,
  Sudan, M., Terzopoulos, D., Tygar, D., Vardi, M.Y., Weikum, G., Esparza, J.,
  Majumdar, R. (eds.) Tools and {{Algorithms}} for the {{Construction}} and
  {{Analysis}} of {{Systems}}, vol.~6015, pp. 150--153. {Springer Berlin
  Heidelberg}, {Berlin, Heidelberg} (2010). \doi{10.1007/978-3-642-12002-2_12}

\bibitem{brzozowskiMorganBisemilattices2000}
Brzozowski, J.: De {{Morgan}} bisemilattices. In: Proceedings 30th {{IEEE
  International Symposium}} on {{Multiple}}-{{Valued Logic}} ({{ISMVL}} 2000).
  pp. 173--178 (May 2000). \doi{10.1109/ISMVL.2000.848616}

\bibitem{AlogtimeIsomorphism1997}
Buss, S.R.: Alogtime algorithms for tree isomorphism, comparison, and
  canonization. In: Gottlob, G., Leitsch, A., Mundici, D. (eds.) Computational
  Logic and Proof Theory. pp. 18--33. Springer Berlin Heidelberg, Berlin,
  Heidelberg (1997)

\bibitem{Cook10.1145/800157.805047}
Cook, S.A.: The complexity of theorem-proving procedures. In: Proceedings of
  the Third Annual ACM Symposium on Theory of Computing. p. 151–158. STOC
  '71, Association for Computing Machinery, New York, NY, USA (1971).
  \doi{10.1145/800157.805047}

\bibitem{DPLL}
Davis, M., Logemann, G., Loveland, D.: A machine program for theorem-proving.
  Commun. ACM  \textbf{5}(7),  394–397 (Jul 1962).
  \doi{10.1145/368273.368557}

\bibitem{ganzingerDPLLFastDecision2004}
Ganzinger, H., Hagen, G., Nieuwenhuis, R., Oliveras, A., Tinelli, C.:
  {{DPLL}}({{T}}): Fast {{Decision Procedures}}. In: Kanade, T., Kittler, J.,
  Kleinberg, J.M., Mattern, F., Mitchell, J.C., Naor, M., Nierstrasz, O.,
  Pandu~Rangan, C., Steffen, B., Sudan, M., Terzopoulos, D., Tygar, D., Vardi,
  M.Y., Weikum, G., Alur, R., Peled, D.A. (eds.) Computer {{Aided
  Verification}}, vol.~3114, pp. 175--188. {Springer Berlin Heidelberg},
  {Berlin, Heidelberg} (2004). \doi{10.1007/978-3-540-27813-9_14}

\bibitem{Gentzen1935}
Gentzen, G.: Untersuchungen \"{u}ber das logische schließen. {I}.
  Mathematische Zeitschrift  \textbf{39},  176--210 (1935)

\bibitem{HamzaETAL19SystemFR}
Hamza, J., Voirol, N., Kun\v{c}ak, V.: {System FR}: Formalized foundations for
  the {Stainless} verifier. Proc. ACM Program. Lang  \textbf{3} (November
  2019). \doi{https://doi.org/10.1145/3360592}

\bibitem{harrisonHOLLightOverview2009}
Harrison, J.: {{HOL Light}}: An {{Overview}}. In: Berghofer, S., Nipkow, T.,
  Urban, C., Wenzel, M. (eds.) Theorem {{Proving}} in {{Higher Order Logics}},
  vol.~5674, pp. 60--66. {Springer Berlin Heidelberg}, {Berlin, Heidelberg}
  (2009). \doi{10.1007/978-3-642-03359-9_4}

\bibitem{hopcroftDesignAnalysisComputer}
Hopcroft, J., UIIman, J., Aho, A.: The {{Design And Analysis Of Computer
  Algorithms}}. Addison-Wesley (1974)

\bibitem{huntComputationalComplexityAlgebra1987}
Hunt, H.~B., I., Rosenkrantz, D.J., Bloniarz, P.A.: On the {{Computational
  Complexity}} of {{Algebra}} on {{Lattices}}. SIAM Journal on Computing
  \textbf{16}(1),  129--148 (Feb 1987). \doi{10.1137/0216011}

\bibitem{kahnTopologicalSortingLarge1962}
Kahn, A.B.: Topological sorting of large networks. Communications of the ACM
  \textbf{5}(11),  558--562 (Nov 1962). \doi{10.1145/368996.369025}

\bibitem{kalmbachOrthomodularLattices1983}
Kalmbach, G.: {Orthomodular Lattices}. {Academic Press Inc}, {London ; New
  York} (Mar 1983)

\bibitem{ProofComplexityPudlak2019}
Kraj\'i\v{c}ek, J.: Proof Complexity. Encyclopedia of Mathematics and Its
  Appplications, Vol.170, Cambridge University Press (2019)

\bibitem{KroeningStrichman}
Kroening, D., Strichman, O.: Decision Procedures - An Algorithmic Point of
  View. Springer (2016)

\bibitem{Kuncak07DecisionProceduresModularDataStructureVerification}
Kuncak, V.: Modular Data Structure Verification. Ph.D. thesis, EECS Department,
  Massachusetts Institute of Technology (February 2007),
  \url{http://hdl.handle.net/1721.1/38533}

\bibitem{DafnyCalculations2014}
Leino, K.R.M., Polikarpova, N.: Verified calculations. In: Cohen, E.,
  Rybalchenko, A. (eds.) Verified Software: Theories, Tools, Experiments. pp.
  170--190. Springer Berlin Heidelberg, Berlin, Heidelberg (2014)

\bibitem{lewisHazardDetectionQuinary1972}
Lewis, D.W.: Hazard detection by a quinary simulation of logic devices with
  bounded propagation delays. In: Proceedings of the 9th {{Design Automation
  Workshop}}. pp. 157--164. {{DAC}} '72, {Association for Computing Machinery},
  {New York, NY, USA} (Jun 1972). \doi{10.1145/800153.804941}

\bibitem{LindellTreesLogspace1992}
Lindell, S.: A logspace algorithm for tree canonization (extended abstract).
  In: Proceedings of the Twenty-Fourth Annual ACM Symposium on Theory of
  Computing. p. 400–404. STOC '92, Association for Computing Machinery, New
  York, NY, USA (1992). \doi{10.1145/129712.129750}

\bibitem{mcallesterAutomaticRecognitionTractability1993}
McAllester, D.A.: Automatic recognition of tractability in inference relations
  \textbf{40}(2),  284--303. \doi{10.1145/151261.151265}

\bibitem{meinanderSolutionUniformWord2010}
Meinander, A.: A solution of the uniform word problem for ortholattices.
  Mathematical Structures in Computer Science  \textbf{20}(4),  625--638 (Aug
  2010). \doi{10.1017/S0960129510000125}

\bibitem{merzAutomaticVerificationTLA2012}
Merz, S., Vanzetto, H.: Automatic {{Verification}} of {{TLA}}\,+\, {{Proof
  Obligations}} with {{SMT Solvers}}. In: Bj{\o}rner, N., Voronkov, A. (eds.)
  Logic for {{Programming}}, {{Artificial Intelligence}}, and {{Reasoning}}.
  pp. 289--303. Lecture {{Notes}} in {{Computer Science}}, {Springer}, {Berlin,
  Heidelberg} (2012). \doi{10.1007/978-3-642-28717-6_23}

\bibitem{naumowiczBriefOverviewMizar2009}
Naumowicz, A., Korni{\l}owicz, A.: A brief overview of mizar. In: Berghofer,
  S., Nipkow, T., Urban, C., Wenzel, M. (eds.) Theorem Proving in Higher Order
  Logics. pp. 67--72. Springer Berlin Heidelberg, Berlin, Heidelberg (2009)

\bibitem{PetersonStickel1981}
Peterson, G.E., Stickel, M.E.: Complete sets of reductions for some equational
  theories. J. ACM  \textbf{28}(2),  233–264 (Apr 1981).
  \doi{10.1145/322248.322251}

\bibitem{pudlakLengthsProofs1998}
Pudl{\'a}k, P.: The {{Lengths}} of {{Proofs}}. In: Studies in {{Logic}} and the
  {{Foundations}} of {{Mathematics}}, vol.~137, pp. 547--637. {Elsevier}
  (1998). \doi{10.1016/S0049-237X(98)80023-2}

\bibitem{AutoActive2015}
Tschannen, J., Furia, C.A., Nordio, M., Polikarpova, N.: Autoproof: Auto-active
  functional verification of object-oriented programs. In: Baier, C., Tinelli,
  C. (eds.) Tools and Algorithms for the Construction and Analysis of Systems.
  pp. 566--580. Springer Berlin Heidelberg, Berlin, Heidelberg (2015)

\bibitem{Urquhart10.1145/7531.8928}
Urquhart, A.: Hard examples for resolution. J. ACM  \textbf{34}(1),  209–219
  (Jan 1987). \doi{10.1145/7531.8928}

\bibitem{wenzelIsabelleFramework2008}
Wenzel, M., Paulson, L.C., Nipkow, T.: The {{Isabelle Framework}}. In: Mohamed,
  O.A., Mu{\~n}oz, C., Tahar, S. (eds.) Theorem {{Proving}} in {{Higher Order
  Logics}}. pp. 33--38. Lecture {{Notes}} in {{Computer Science}}, {Springer},
  {Berlin, Heidelberg} (2008). \doi{10.1007/978-3-540-71067-7_7}

\bibitem{whitmanFreeLattices1941}
Whitman, P.M.: Free {{Lattices}}. Annals of Mathematics  \textbf{42}(1),
  325--330 (1941). \doi{10.2307/1969001}

\bibitem{ZeeETAL08FullFunctionalVerificationofLinkedDataStructures}
Zee, K., Kuncak, V., Rinard, M.: Full functional verification of linked data
  structures. In: ACM SIGPLAN Conf. Programming Language Design and
  Implementation (PLDI) (2008), see also
  \cite{Kuncak07DecisionProceduresModularDataStructureVerification}

\bibitem{ZeeETAL09IntegratedProofLanguageforImperativePrograms}
Zee, K., Kuncak, V., Rinard, M.: An integrated proof language for imperative
  programs. In: ACM SIGPLAN Conf. Programming Language Design and
  Implementation (PLDI) (2009)

\end{thebibliography}
}

\vfill

{\small\medskip\noindent{\bf Open Access} This chapter is licensed under the terms of the Creative Commons\break Attribution 4.0 International License (\url{http://creativecommons.org/licenses/by/4.0/}), which permits use, sharing, adaptation, distribution and reproduction in any medium or format, as long as you give appropriate credit to the original author(s) and the source, provide a link to the Creative Commons license and indicate if changes were made.}

{\small \spaceskip .28em plus .1em minus .1em The images or other
third party material in this chapter are included in the\break
chapter's Creative Commons license, unless indicated otherwise in a
credit line to the\break material.~If material is not included in
the chapter's Creative Commons license and\break your intended use
is not permitted by statutory regulation or exceeds the
permitted\break use, you will need to obtain permission directly
from the copyright holder.}

\medskip\noindent\includegraphics{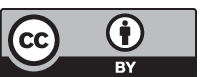}

\end{document}